\documentclass[oneside,english,american]{amsart}
\usepackage[T1]{fontenc}
\usepackage[latin9]{inputenc}
\usepackage[a4paper]{geometry}
\usepackage{amsthm}
\usepackage{amstext}

\makeatletter

\newcommand{\lyxmathsym}[1]{\ifmmode\begingroup\def\b@ld{bold}
  \text{\ifx\math@version\b@ld\bfseries\fi#1}\endgroup\else#1\fi}

\numberwithin{equation}{section} 
\numberwithin{figure}{section} 
\theoremstyle{plain}
\newtheorem{thm}{Theorem}
  \theoremstyle{remark}
  \newtheorem{rem}[thm]{Remark}

\makeatother

\usepackage{babel}

\begin{document}

\title{Old and new results about relativistic Hermite polynomials}

\author{C. Vignat}

\maketitle

\section{Introduction}

The relativistic Hermite polynomials (RHP) were introduced in 1991
by Aldaya et al. \cite{Aldaya} in a generalization of the theory
of the quantum harmonic oscillator to the relativistic context. These
polynomials were later related to the more classical Gegenbauer (or
ultraspherical) polynomials in a study by Nagel \cite{Nagel}. Thus
some of their properties can be deduced from the properties of the
well-known Gegenbauer polynomials, as underlined by M. Ismail in \cite{Ismail}.
In this report we give new proofs of already known results but also
new results about these polynomials. We use essentially three basic
tools: the representation of polynomials as moments, the subordination
tool and Nagel's identity.

\section{Definitions and Tools}

\subsection{Polynomials}

\subsubsection{Relativistic Hermite polynomials}

\begin{flushleft}
The relativistic Hermite polynomial of degree $n$ and parameter $N\ne0$
is defined by the Rodrigues formula\foreignlanguage{english}{\[
H_{n}^{N}\left(X\right)=\left(-1\right)^{n}\left(1+\frac{X^{2}}{N}\right)^{N+n}\frac{d^{n}}{dX^{n}}\left(1+\frac{X^{2}}{N}\right)^{-N};\]
}examples of RHP polynomials are\begin{eqnarray*}
H_{0}^{N}\left(X\right)=1;H_{1}^{N}\left(X\right)=2X;\,\, H_{2}^{N}\left(X\right)=2\left(-1+X^{2}\left(2+\frac{1}{N}\right)\right)\\
H_{3}^{N}\left(X\right)=4\left(1+\frac{1}{N}\right)\left(X^{3}\left(2+\frac{1}{N}\right)-3X\right)\end{eqnarray*}
These polynomials are extensions of the classical Hermite polynomials
$H_{n}\left(X\right)$ that are defined by the Rodrigues formula\[
H_{n}\left(X\right)=\left(-1\right)^{n}\exp\left(X^{2}\right)\frac{d^{n}}{dX^{n}}\exp\left(-X^{2}\right)\]
and thus can be obtained as the limit case\[
\lim_{N\to+\infty}H_{n}^{N}\left(X\right)=H_{n}\left(X\right).\]
 An explicit formula for the relativistic Hermite polynomial is \cite{Aldaya}\begin{equation}
H_{n}^{N}\left(X\right)=\frac{\left(2N\right)_{n}}{\left(2\sqrt{N}\right)^{n}}\sum_{k=0}^{\left\lfloor \frac{n}{2}\right\rfloor }\frac{\left(-1\right)^{k}}{\left(N+\frac{1}{2}\right)_{k}}\frac{n!}{\left(n-2k\right)!k!}\left(\frac{2X}{\sqrt{N}}\right)^{n-2k}\label{eq:explicitRHP}\end{equation}
where $\left(n\right)_{k}$ is the Pochhammer symbol.
\par\end{flushleft}

\subsubsection{Gegenbauer polynomials}

\begin{flushleft}
The Gegenbauer polynomial of degree $n$ and parameter $N$ is defined
by the Rodrigues formula\[
C_{n}^{N}\left(X\right)=\gamma_{n}^{N}\left(-1\right)^{n}\left(1-X^{2}\right)^{\frac{1}{2}-N}\frac{d^{n}}{dX^{n}}\left(1-X^{2}\right)^{n+N-\frac{1}{2}}\]
with \[
\gamma_{n}^{N}=\frac{\left(2N\right)_{n}}{2^{n}n!\left(N+\frac{1}{2}\right)_{n}};\]
examples of Gegenbauer polynomials are\foreignlanguage{english}{\[
C_{0}^{N}\left(X\right)=1;\lyxmathsym{ }\,\, C\left(X\right)=2NX;\lyxmathsym{ }\,\, C\left(X\right)=2N\left(N+1\right)X^{2}-N;\,\, C_{3}^{N}\left(X\right)=2N\left(N+1\right)\left(\frac{2\left(N+2\right)}{3}X^{3}-X\right)\]
}and an explicit formula is\begin{equation}
C_{n}^{N}\left(X\right)=\sum_{k=0}^{\left\lfloor \frac{n}{2}\right\rfloor }\left(-1\right)^{k}\frac{\left(N\right)_{n-k}}{\left(n-2k\right)!k!}\left(2X\right)^{n-2k}.\label{eq:explicitGegenbauer}\end{equation}

\par\end{flushleft}

\subsection{tools}

\subsubsection{Nagel's identity}

Nagel's identity \cite{Nagel} is a first link between relativistic
Hermite polynomials and Gegenbauer polynomials. The RHP with degree
$n$ and parameter $N$ is linked to the Gegenbauer polynomial with
same degree and same parameter via Nagel's identity\foreignlanguage{english}{\[
H_{n}^{N}\left(X\sqrt{N}\right)=\frac{n!}{N^{\frac{n}{2}}}\left(1+X^{2}\right)^{\frac{n}{2}}C_{n}^{N}\left(\frac{X}{\sqrt{1+X^{2}}}\right).\]
}As a consequence of Nagel's identity, we deduce the following theorem
\begin{thm}
The relativistic Hermite polynomials is related to the Gegenbauer
polynomial as\begin{equation}
C_{n}^{N}\left(X\right)=\alpha_{n}^{N}H_{n}^{\frac{1}{2}-N-n}\left(-iX\sqrt{\frac{1}{2}-N-n}\right)\label{eq:CniX}\end{equation}

with\[
\alpha_{n}^{N}=\left(-2i\right)^{n}\frac{\left(N\right)_{n}}{\left(2N+n\right)_{n}}\frac{\left(\frac{1}{2}-N-n\right)^{\frac{n}{2}}}{n!}.\]
\end{thm}
\begin{proof}
By the identity \cite[(7.2.15.8)]{Brychkov},\[
C_{n}^{N}\left(iX\right)=\left(-2i\right)^{n}\frac{\left(N\right)_{n}}{\left(2N+n\right)_{n}}\left(1+X^{2}\right)^{\frac{n}{2}}C_{n}^{\frac{1}{2}-N-n}\left(\frac{X}{\sqrt{1+X^{2}}}\right)\]
and the result follows by application of Nagel's identity.
\end{proof}
Formula (\ref{eq:CniX}) was derived by Nagel \cite{Nagel}, who notes:
{}``This representation does not seem to be very useful however.''
Indeed, the fact that the parameter $N$ of the Gegenbauer polynomial
is transformed, for the Relativistic Hermite polynomial, into a parameter
$\frac{1}{2}-N-n$ that depends on the degree $n$, seems to make
this identity a priori less useful than expected. However, as will
be shown in Section \ref{sec:scaling}, this formula allows in fact
an easy extension of the scaling identity for Gegenbauer polynomials
to the case of Relativistic Hermite polynomials. 

A similar version of formula (\ref{eq:CniX}) appears in \cite{Midya}
and also in \cite{Masjedjamei} in the case of deformed Hermite polynomials.

\subsubsection{The subordination tool}

Writing polynomials as scale mixtures of others allows to deduce properties
between them. We use the probabilistic expectation to denote the subordination
by the mesure $f$ as \[
E_{b}H_{n}\left(X\sqrt{b}\right)=\int_{0}^{+\infty}H_{n}\left(X\sqrt{b}\right)f\left(b\right)db.\]
The subordination dependence between Hermite, relativistic Hermite
and Gegenbauer polynomials is as follows.
\begin{thm}
The Hermite, relativistic Hermite and Gegenbauer polynomials are related
as follows

\[
H_{n}\left(X\right)=\frac{N^{\frac{n}{2}}}{\left(N\right)_{\frac{n}{2}}}E_{c}H_{n}^{N}\left(\frac{X\sqrt{N}}{\sqrt{c}}\right),\,\,\, c\sim\Gamma\left(N+\frac{n+1}{2}\right).\]
\begin{equation}
C_{n}^{N}\left(X\right)=\frac{\left(N\right)_{\frac{n}{2}}}{n!}E_{b}H_{n}\left(X\sqrt{b}\right),\,\,\, b\sim\Gamma\left(N+\frac{n}{2}\right)\label{eq:Gegenbauersubordination}\end{equation}

Here, $b\sim\Gamma\left(N+\frac{n}{2}\right)$ means that $b$ is
a random variable that follows the Gamma distribution with shape parameter
$N+\frac{n}{2}$ that is $\gamma_{N+\frac{n}{2}}\left(b\right)=\frac{1}{\Gamma\left(N+\frac{n}{2}\right)}e^{-b}b^{N+\frac{n}{2}-1},\,\, b\ge0$.\end{thm}
\begin{proof}
Since\[
Ec^{l}=\left(N+\frac{n+1}{2}\right)_{l}\]
we deduce from (\ref{eq:explicitRHP}) that \begin{eqnarray*}
E_{c}H_{n}^{N}\left(\frac{X\sqrt{N}}{\sqrt{c}}\right) & = & \frac{\left(2N\right)_{n}}{\left(2\sqrt{N}\right)^{n}}\sum_{k=0}^{\left\lfloor \frac{n}{2}\right\rfloor }\frac{\left(-1\right)^{k}}{\left(N+\frac{1}{2}\right)_{k}}\frac{n!}{\left(n-2k\right)!k!}\left(N+\frac{n+1}{2}\right)_{k-\frac{n}{2}}\left(2X\right)^{n-2k}\\
 & = & \frac{\left(2N\right)_{n}}{\left(2\sqrt{N}\right)^{n}\left(N+\frac{1}{2}\right)_{\frac{n}{2}}}H_{n}\left(X\right)\end{eqnarray*}
so that the first result follows after application of Euler's duplication
formula.

The same way, we compute\begin{eqnarray*}
\frac{\left(N\right)_{\frac{n}{2}}}{n!}E_{b}H_{n}\left(X\sqrt{b}\right) & = & \sum_{k=0}^{\left\lfloor \frac{n}{2}\right\rfloor }\left(-1\right)^{k}\frac{\left(N\right)_{\frac{n}{2}}}{\left(n-2k\right)!k!}\left(N+\frac{n}{2}\right)_{\frac{n}{2}-k}\left(2X\right)^{n-2k}\\
 & = & \sum_{k=0}^{\left\lfloor \frac{n}{2}\right\rfloor }\left(-1\right)^{k}\frac{\left(N\right)_{n-k}}{\left(n-2k\right)!k!}\left(2X\right)^{n-2k}\end{eqnarray*}
which coincides with (\ref{eq:explicitGegenbauer}).
\end{proof}

\subsubsection{The moment representation}

The well-known moment representation of the Hermite polynomial \begin{equation}
H_{n}\left(X\right)=2^{n}E_{Z}\left(X+iZ\right)^{n},\label{eq:momentrepresentationHermite}\end{equation}
where $Z$ is Gaussian centered with variance $\frac{1}{2}$ is a
consequence of the integral formula \cite[8.951]{Gradshteyn}. 

Its extension to the Gegenbauer polynomials is deduced from integral
\cite[8.931]{Gradshteyn} (called Laplace first integral in \cite{Ismail})
and reads\foreignlanguage{english}{\begin{equation}
C_{n}^{N}\left(X\right)=\frac{\left(2N\right)_{n}}{n!}E_{Z_{N}}\left[X+i\sqrt{1-X^{2}}Z_{N}\right]^{n}\label{eq:momentCZN}\end{equation}
}where the random variable $Z_{N}$ has a Student-r distribution\begin{equation}
f_{Z_{N}}\left(Z\right)=\frac{\Gamma\left(N+\frac{1}{2}\right)}{\Gamma\left(N\right)\Gamma\left(\frac{1}{2}\right)}\left(1-Z^{2}\right)^{N-1},\lyxmathsym{  }-1\le Z\le+1.\label{eq:fZn}\end{equation}
Using Nagel's identity, we deduce from (\ref{eq:momentCZN}) the moment
representation for RHP polynomials\begin{equation}
H_{n}^{N}\left(X\right)=\frac{\left(2N\right)_{n}}{N^{\frac{n}{2}}}E_{Z_{N}}\left(\frac{X}{\sqrt{N}}+iZ_{N}\right)^{n}\label{eq:RHPmoment}\end{equation}
with $Z_{N}$ distributed according to (\ref{eq:fZn}). 

Another moment representation for the Gegenbauer polynomials is as
follows:
\begin{thm}
The Gegenbauer polynomial and relativistic Hermite polynomials have
moment representation

\begin{equation}
C_{n}^{N}\left(X\right)=\frac{2^{n}\left(N\right)_{\frac{n}{2}}}{n!}E_{Z,b}\left(X\sqrt{b}+iZ\right)^{n}\label{eq:CbZ}\end{equation}
and\begin{equation}
H_{n}^{N}\left(X\sqrt{N}\right)=\frac{2^{n}\left(2N\right)_{n}}{N^{\frac{n}{2}}}E_{b,Z}\left(X\sqrt{b}+i\sqrt{1+X^{2}}Z\right)^{n}\label{eq:HNnbZ}\end{equation}
where $Z$ is Gaussian centered with variance $\frac{1}{2}$ and independent
of $b$ which is Gamma distributed with shape parameter $N+\frac{n}{2}.$\end{thm}
\begin{proof}
This expression is derived by the application of the subordination
identity (\ref{eq:Gegenbauersubordination}) to the moment representation
(\ref{eq:momentrepresentationHermite}). The representation (\ref{eq:HNnbZ})
is deduced from (\ref{eq:CbZ}) using Nagel's identity.
\end{proof}
However, another set of moment representations involving two random
variables can be deduced from (\ref{eq:momentCZN}) as follows:
\begin{thm}
The Gegenbauer polynomial has for moment representation \begin{equation}
C_{n}^{N}\left(X\right)=\frac{1}{n!}E\left[\left(X+\sqrt{X^{2}-1}\right)U+\left(X-\sqrt{X^{2}-1}\right)V\right]^{n}\label{eq:SunGegenbauer}\end{equation}
where $U$ and $V$ are independently distributed according to a Gamma
law with shape parameter $N.$ \end{thm}
\begin{proof}
Consider $U$ and $V$ independently distributed according to a Gamma
law with shape parameter $N;$ then $U+V$ is Gamma distributed with
shape parameter $2N$ so that $E\left(U+V\right)^{n}=\left(2N\right)_{n}$.
We deduce from (\ref{eq:momentCZN}) that\[
C_{n}^{N}\left(X\right)=\frac{E\left(U+V\right)^{n}}{n!}E_{Z_{N}}\left[X+i\sqrt{1-X^{2}}Z_{N}\right]^{n}=\frac{1}{n!}E\left[X\left(U+V\right)+i\sqrt{1-X^{2}}Z_{N}\left(U+V\right)\right]^{n}\]
but a well-known stochatic representation for $Z_{N}$ is\begin{equation}
Z_{N}=\frac{U-V}{U+V}\label{eq:stochasticZN}\end{equation}
where $Z_{N}$ is independent of $\left(U+V\right)$ so that \[
C_{n}^{N}\left(X\right)=\frac{1}{n!}E\left[X\left(U+V\right)+i\sqrt{1-X^{2}}\left(U-V\right)\right]^{n}\]
and the result follows. 
\end{proof}
From this result we deduce 
\begin{thm}
A moment representation for the relativistic Hermite polynomial is\begin{equation}
H_{n}^{N}\left(X\sqrt{N}\right)=\frac{1}{N^{\frac{n}{2}}}E\left[\left(i+X\right)U+\left(-i+X\right)V\right]^{n}\label{eq:momentHermiteUV}\end{equation}
where $U$ and V are independently distributed according to a Gamma
law with shape parameter $N.$ \end{thm}
\begin{proof}
This is a direct consequence of Nagel's formula.\end{proof}
\begin{rem}
Representation (\ref{eq:RHPmoment}) can proved directly from the
explicit expression (\ref{eq:explicitRHP}) of the Relativistic Hermite
polynomials: since $Z_{N}$ has odd moments equal to zero, we have\[
E_{Z_{N}}\left(\frac{X}{\sqrt{N}}+iZ_{N}\right)^{n}=\sum_{k=0}^{\left\lfloor \frac{n}{2}\right\rfloor }\binom{n}{2k}\left(\frac{X}{\sqrt{N}}\right)^{n-2k}E\left(iZ_{N}\right)^{2k};\]
the even moments can be computed as \[
EZ_{N}^{2k}=\frac{\Gamma\left(k+\frac{1}{2}\right)\Gamma\left(N+\frac{1}{2}\right)}{\Gamma\left(N+k+\frac{1}{2}\right)\Gamma\left(\frac{1}{2}\right)}=\frac{2k!}{k!2^{2k}}\frac{1}{\left(N+\frac{1}{2}\right)_{k}}\]
so that\[
\frac{\left(2N\right)_{n}}{N^{n/2}}E_{Z_{N}}\left(X+iZ_{N}\right)^{n}=\frac{\left(2N\right)_{n}}{\left(2\sqrt{N}\right)^{n}}\sum_{k=0}^{\left\lfloor \frac{n}{2}\right\rfloor }\frac{\left(-1\right)^{k}}{\left(N+\frac{1}{2}\right)_{k}}\frac{n!}{\left(n-2k\right)!k!}\left(\frac{2X}{\sqrt{N}}\right)^{n-2k}\]
which coincides with (\ref{eq:explicitRHP}).

We note moreover that the moment representation (\ref{eq:SunGegenbauer})
was derived recently by Sun \cite{Sun} using another proof based
on the generating function of the Gegenbauer polynomials. 
\end{rem}

\subsubsection{example of application}

As an example of the usefulness of the moment representations given
above, we derive the famous \cite[8.952.1]{Gradshteyn} \[
\frac{d}{dX}H_{n}\left(X\right)=\frac{d}{dX}2^{n}E\left(X+iZ\right)^{n}=n2^{n}E\left(X+iZ\right)^{n-1}=2nH_{n-1}\left(X\right),\]
and its relativistic version\begin{eqnarray*}
\frac{d}{dX}H_{n}^{N}\left(X\right) & = & \frac{d}{dX}\frac{\left(2N\right)_{n}}{N^{n/2}}E_{Z_{N}}\left(\frac{X}{\sqrt{N}}+iZ_{N}\right)^{n}\\
 & = & \frac{\left(2N\right)_{n}}{N^{n/2}}\frac{n}{\sqrt{N}}E_{Z_{N}}\left(\frac{X}{\sqrt{N}}+iZ_{N}\right)^{n-1}\\
 & = & \frac{n\left(2N+n-1\right)}{N}H_{n-1}^{N}\left(X\right).\end{eqnarray*}
In the Gegenbauer case, we rather use the stochastic representation
(\ref{eq:Gegenbauersubordination}) to obtain \[
\frac{d}{dX}C_{n}^{N}\left(X\right)=\frac{\left(N\right)_{\frac{n}{2}}}{n!}E_{b\sim\Gamma_{N+\frac{n}{2}}}\frac{d}{dX}\left(H_{n}\left(X\sqrt{b}\right)\right)=\frac{\left(N\right)_{\frac{n}{2}}}{n!}2nE_{b\sim\Gamma_{N+\frac{n}{2}}}\sqrt{b}H_{n-1}\left(X\sqrt{b}\right)\]
and since $E_{b\sim\Gamma_{N+\frac{n}{2}}}\sqrt{b}f\left(b\right)=\frac{\Gamma\left(N+\frac{n}{2}+\frac{1}{2}\right)}{\Gamma\left(N+\frac{n}{2}\right)}E_{c\sim\Gamma_{N+\frac{n+1}{2}}}f\left(c\right),$
we deduce\[
\frac{d}{dX}C_{n}^{N}\left(X\right)=\frac{\left(N\right)_{\frac{n}{2}}}{n!}2n\frac{\Gamma\left(N+\frac{n}{2}+\frac{1}{2}\right)}{\Gamma\left(N+\frac{n}{2}\right)}\frac{\left(n-1\right)!}{\left(N+1\right)_{\frac{n-1}{2}}}C_{n-1}^{N+1}\left(X\right)=2NC_{n-1}^{N+1}\left(X\right)\]
which coincides with \cite[8.935.2]{Gradshteyn}.

\section{The Gram-Schmidt operator}

A family of orthogonal polynomials can be obtained by applying the
Gram-Schmidt operator to the canonical basis $\left\{ 1,X,\dots,X^{n}\right\} .$
We show here how this operator can be expressed in the case where
a moment formula exists.
\begin{thm}
If a polynomial $P_{n}\left(X\right)$ can be expressed as \[
P_{n}\left(X\right)=E\left[X+iZ\right]^{n}\]
for some random variable $Z$ then\[
P_{n}\left(X\right)=\phi_{Z}\left(\frac{d}{dX}\right)X^{n}\]
where $\phi_{Z}\left(u\right)=E_{Z}\exp\left(iuZ\right)$ is the characteristic
function of $Z.$\end{thm}
\begin{proof}
By definition\begin{eqnarray*}
E\left[X+iZ\right]^{n}=\sum_{k=0}^{n}\binom{n}{k}i^{k}EZ^{k}X^{n-k}=\sum_{k=0}^{+\infty}\frac{i^{k}}{k!}EZ^{k}\frac{d^{k}}{dX^{k}}X^{n}\\
=E_{Z}\exp\left(iZ\frac{d}{dX}\right)X^{n} & =\phi_{Z}\left(\frac{d}{dX}\right)X^{n}\end{eqnarray*}

\end{proof}
As an application of this theorem, we recover the following well-known
result for the Hermite polynomials\[
H_{n}\left(X\right)=\exp\left(-\frac{1}{4}\frac{d^{2}}{dX^{2}}\right)\left(2X\right)^{n}.\]
The extension of this result to the case of the Relativistic Hermite
polynomials is as follows
\begin{thm}
The Gram-Schmidt operator associated to the relativistic Hermite polynomial
is \[
H_{n}^{N}\left(X\sqrt{N}\right)=\frac{\left(2N\right)_{n}}{N^{\frac{n}{2}}}j_{N+\frac{1}{2}}\left(\frac{d}{dX}\right)X^{n}\]
where the normalized Bessel function is \begin{equation}
j_{N+\frac{1}{2}}\left(u\right)=2^{N+\frac{1}{2}}\Gamma\left(N+\frac{3}{2}\right)\frac{J_{N+\frac{1}{2}}\left(u\right)}{u^{N+\frac{1}{2}}}=\sum_{k=0}^{+\infty}\frac{\left(-1\right)^{k}}{k!\left(N+\frac{3}{2}\right)_{k}}\left(\frac{z}{2}\right)^{2k}.\label{eq:Besselj}\end{equation}
\end{thm}
\begin{proof}
The characteristic function of the random variable $Z_{N}$ in (\ref{eq:RHPmoment})
is the normalized Bessel function \[
\phi_{Z_{N}}\left(u\right)=j_{N+\frac{1}{2}}\left(u\right)\]

\end{proof}
As $N\to+\infty,$ since \[
\lim_{N\to+\infty}\frac{\left(2N\right)_{n}}{N^{\frac{n}{2}}}j_{N+\frac{1}{2}}\left(u\right)=\exp\left(-\frac{u^{2}}{4}\right),\]
we recover the classical Hermite case.

\section{addition theorems}

\subsection{A new proof of the classical addition theorem for Hermite polynomials}

The summation theorem for Hermite polynomials \cite[8.958.1]{Gradshteyn}
states that\begin{equation}
\frac{\left(\sum_{k=1}^{r}a_{k}^{2}\right)^{\frac{n}{2}}}{n!}H_{n}\left(\frac{\sum_{k=1}^{r}a_{k}X_{k}}{\sqrt{\sum_{k=1}^{r}a_{k}^{2}}}\right)=\sum_{m_{1}+\dots+m_{r}=n}\prod_{k=1}^{r}\frac{a_{k}^{m_{k}}}{m_{k}!}H_{m_{k}}\left(X_{k}\right)\label{eq:summationHermite}\end{equation}
We give here a short proof using the moment representation; we assume
first that $\sum_{k=1}^{r}a_{k}^{2}=1$ so that we expand\[
H_{n}\left(\sum_{k=1}^{r}a_{k}X_{k}\right)=2^{2n}E_{Z}\left[\sum_{k=1}^{r}a_{k}X_{k}+iZ\right]^{n}=2^{2n}E_{Z_{1},\dots,Z_{r}}\left[\sum_{k=1}^{r}a_{k}\left(X_{k}+iZ_{k}\right)\right]^{n}\]
where variables $Z_{k}$ are independent and Gaussian with variance
$\frac{1}{2}$ so that $\sum_{k=1}^{r}a_{k}Z_{k}$ is Gaussian with
variance $\frac{1}{2}$ and we deduce\[
H_{n}\left(\sum_{k=1}^{r}a_{k}X_{k}\right)=2^{2n}n!\sum_{m_{1}+\dots+m_{r}=n}\prod_{k=1}^{r}\frac{a_{k}^{m_{k}}2^{-2m_{k}}}{m_{k}!}H_{m_{k}}\left(X_{k}\right)\]
so that the result follows. Now we replace $a_{k}$ by $\frac{a_{k}}{\sqrt{\sum_{l=1}^{r}a_{l}^{2}}}$
in this equality so that we obtain the general case (\ref{eq:summationHermite}).

\subsection{An addition theorem for the relativistic Hermite polynomial}

From the well-known addition formula \cite[7.2.13.36]{Brychkov} for
Gegenbauer polynomials\[
C_{n}^{N}\left(X+Y\right)=\sum_{k=0}^{n}\frac{\left(N\right)_{n-k}}{\left(n-k\right)!}\left(2X\right)^{n-k}C_{k}^{N+n-k}\left(Y\right)\]
we deduce the following
\begin{thm}
An addition theorem for the relaivistic Hermite polynomials is\[
\tilde{H}_{n}^{N}\left(X+Y\right)=\sum_{k=0}^{n}\binom{n}{k}\left(1-2N-n\right)_{n-k}\tilde{H}_{k}^{N}\left(Y\right)\]
\end{thm}
\begin{proof}
Using formula (\ref{eq:CniX}), we deduce\[
\alpha_{n}^{N}\tilde{H}_{n}^{\frac{1}{2}-n-N}\left(-iX-iY\right)=\sum_{k=0}^{n}\frac{\left(N\right)_{n-k}}{\left(n-k\right)!}\left(2X\right)^{n-k}\alpha_{k}^{N+n-k}\tilde{H}_{k}^{\frac{1}{2}-N-n+k}\left(-iY\right).\]
Replacing $X$ by $iX$ and $Y$ by $iX$ and computing\begin{eqnarray*}
\frac{\alpha_{k}^{N+n-k}}{\alpha_{n}^{N}} & = & \left(-2i\right)^{k-n}\frac{n!\left(N+n-k\right)_{k}\left(2N+n\right)_{n}}{k!\left(N\right)_{n}\left(2N+2n-2\right)_{k}}\frac{\left(\frac{1}{2}-N-n\right)^{\frac{k}{2}}}{\left(\frac{1}{2}-N-n\right)^{\frac{n}{2}}}\\
 & = & \left(-2i\right)^{k-n}\frac{n!}{k!}\frac{\Gamma\left(2N+2n-k\right)\Gamma\left(N\right)}{\Gamma\left(N+n-k\right)\Gamma\left(2N+n\right)}\left(\frac{1}{2}-N-n\right)^{\frac{k-n}{2}}\end{eqnarray*}
so that, replacing $X$ by $iX$ and $Y$ by $iY$ yields\[
\tilde{H}_{n}^{\frac{1}{2}-n-N}\left(X+Y\right)=\sum_{k=0}^{n}\binom{n}{k}\left(-\frac{X}{\sqrt{\frac{1}{2}-N-n}}\right)^{n-k}\left(2N+n\right)_{n-k}\tilde{H}_{k}^{\frac{1}{2}-N-n}\left(Y\right)\]
and the result is obtained by replacing $\frac{1}{2}-N-n$ by $N.$
\end{proof}

\section{\label{sec:scaling}The scaling identity}

The scaling identity for Hermite polynomials reads \cite[4.6.33]{Ismail}\begin{equation}
H_{n}\left(cX\right)=\sum_{l=0}^{\left\lfloor \frac{n}{2}\right\rfloor }\left(-1\right)^{l}\frac{n!}{\left(n-2l\right)!l!}\left(1-c^{2}\right)^{l}c^{n-2l}H_{n-2l}\left(X\right)\label{eq:Hermitescaling}\end{equation}
A quick proof can be given using the moment representation (\ref{eq:momentrepresentationHermite}):\begin{eqnarray*}
H_{n}\left(cX\right) & = & 2^{n}E_{Z}\left(cX+iZ\right)^{n}\\
 & = & 2^{n}E_{Z_{1},Z_{2}}\left(cX+icZ_{1}+i\sqrt{1-c^{2}}Z_{2}\right)^{n}\end{eqnarray*}
where $Z_{1}$ and $Z_{2}$ are independent Gaussian random variables
with variance $\frac{1}{2}$ so that\begin{eqnarray*}
H_{n}\left(cX\right) & = & 2^{n}\sum_{k=0}^{n}\binom{n}{k}i^{k}\left(1-c^{2}\right)^{\frac{k}{2}}EZ_{2}^{k}c^{n-k}E\left(X+iZ_{1}\right)^{n-k}\\
 & = & 2^{n}\sum_{l=0}^{\left\lfloor \frac{n}{2}\right\rfloor }\binom{n}{2l}i^{k}\left(1-c^{2}\right)^{l}EZ_{2}^{2l}c^{n-2l}2^{2l-n}H_{n-2l}\left(X\right)\end{eqnarray*}
since the odd moments of a Gaussian are null; as moreover $EZ_{2}^{2l}=2^{2l-1}\frac{\Gamma\left(2l\right)}{\Gamma\left(l\right)}$
we deduce the result.

It is possible to extend this proof to the case of Gegenbauer polynomials
using either the moment representation (\ref{eq:momentCZN}) or (\ref{eq:CbZ});
however, a more simple proof can be derived using the subordination
relation (\ref{eq:Gegenbauersubordination})  as follows.
\begin{thm}
\cite[7.2.13.37]{Brychkov}The scaling identity for Gegenbauer polynomials
reads\begin{equation}
C_{n}^{N}\left(aX\right)=\sum_{l=0}^{\left\lfloor \frac{n}{2}\right\rfloor }\frac{\left(-1\right)^{l}\left(N\right)_{l}}{l!}\left(1-c^{2}\right)^{l}c^{n-2l}C_{n-2l}^{N+l}\left(X\right)\label{eq:Gegenbauerscaling}\end{equation}
\end{thm}
\begin{proof}
From the subordination formula (\ref{eq:Gegenbauersubordination})
and the scaling formula (\ref{eq:Hermitescaling}) we obtain (where
the notation $b_{N+\frac{n}{2}}$ is a shortcut for $b\sim\Gamma_{N+\frac{n}{2}}$)\[
E_{b}H_{n}\left(cX\sqrt{b_{N+\frac{n}{2}}}\right)=\frac{n!}{\left(N\right)_{\frac{n}{2}}}C_{n}^{N}\left(cX\right)=\sum_{l=0}^{\left\lfloor \frac{n}{2}\right\rfloor }\left(-1\right)^{l}\frac{n!}{\left(n-2l\right)!l!}\left(1-c^{2}\right)^{l}c^{n-2l}E_{b}H_{n-2l}\left(X\sqrt{b_{N+\frac{n}{2}}}\right)\]
but\[
EH_{n-2l}\left(X\sqrt{b_{N+\frac{n}{2}}}\right)=EH_{n'}\left(X\sqrt{b_{N'+\frac{n'}{2}}}\right)=\frac{n'!}{\left(N'\right)_{\frac{n'}{2}}}C_{n'}^{N'}\left(X\right)\]
with $n'=n-2l$ and $N'=N+l$ so that\begin{eqnarray*}
C_{n}^{N}\left(cX\right) & = & \frac{\left(N\right)_{\frac{n}{2}}}{n!}\sum_{l=0}^{\left\lfloor \frac{n}{2}\right\rfloor }\left(-1\right)^{l}\frac{n!}{\left(n-2l\right)!l!}\left(1-c^{2}\right)^{l}c^{n-2l}\frac{\left(n-2l\right)!}{\left(N+l\right)_{\frac{n}{2}-l}}C_{n-2l}^{N+l}\left(X\right)\\
 & = & \sum_{l=0}^{\left\lfloor \frac{n}{2}\right\rfloor }\frac{\left(-1\right)^{l}}{l!}\left(N\right)_{l}\left(1-c^{2}\right)^{l}c^{n-2l}C_{n-2l}^{N+l}\left(X\right)\end{eqnarray*}

\end{proof}
The scaling identity for the RHP can be deduced from the preceding
one using formula (\ref{eq:CniX}).
\begin{thm}
The scaling identity for relativistic Hermite polynomials is\[
N^{\frac{n}{2}}H_{n}^{N}\left(cX\sqrt{N}\right)=\sum_{l=0}^{\left\lfloor \frac{n}{2}\right\rfloor }\left(-1\right)^{l}\frac{n!}{\left(n-2l\right)!l!}\left(N\right)_{l}\left(1-c^{2}\right)^{l}c^{n-2l}\left(N+l\right)^{\frac{n-2l}{2}}H_{n-2l}^{N+l}\left(X\sqrt{N+l}\right)\]
\end{thm}
\begin{proof}
Starting from (\ref{eq:Gegenbauerscaling}) and using (\ref{eq:CniX}),
we deduce\[
H_{n}^{\frac{1}{2}-N-n}\left(-icX\sqrt{\frac{1}{2}-N-n}\right)=\sum_{l=0}^{\left\lfloor \frac{n}{2}\right\rfloor }\frac{\left(-1\right)^{l}}{l!}\left(N\right)_{l}\frac{\alpha_{n-2l}^{N+l}}{\alpha_{n}^{N}}\left(1-c^{2}\right)^{l}c^{n-2l}H_{n-2l}^{\frac{1}{2}-N-n+l}\left(-iX\sqrt{\frac{1}{2}-N-n+l}\right)\]
A short computation yields to\[
\left(N\right)_{l}\frac{\alpha_{n-2l}^{N+l}}{\alpha_{n}^{N}}=\left(-1\right)^{l}\frac{n!}{\left(n-2l\right)!}\frac{\left(\frac{1}{2}-N-n+l\right)^{\frac{n-2l}{2}}}{\left(\frac{1}{2}-N-n\right)^{\frac{n}{2}}}\frac{\Gamma\left(N+n+\frac{1}{2}\right)}{\Gamma\left(N+n-l+\frac{1}{2}\right)}.\]
Replacing $N$ by $\frac{1}{2}-n-N$ we deduce\[
H_{n}^{N}\left(cX\sqrt{N}\right)=\sum_{l=0}^{\left\lfloor \frac{n}{2}\right\rfloor }\frac{n!^{l}}{\left(n-2l\right)!l!}\frac{\left(N+l\right)^{\frac{n-2l}{2}}}{N^{\frac{n}{2}}}\frac{\Gamma\left(1-N\right)}{\Gamma\left(1-N-l\right)}\left(1-c^{2}\right)^{l}c^{n-2l}H_{n-2l}^{N+l}\left(X\sqrt{N+l}\right)\]
and the result follows from the fact that\[
\frac{\Gamma\left(1-N\right)}{\Gamma\left(1-N-l\right)}=\left(-1\right)^{l}\left(N\right)_{l}.\]

\end{proof}

\section{Generating Functions}

\subsection{the generating function for the RHP}

The generating function for the RHP is computed in \cite{Dattoli}
using a differential equation; we note that it can not be obtained
directly using formula (\ref{eq:CniX}). However, it can be easily
obtained from the moment representation (\ref{eq:momentHermiteUV})
or (\ref{eq:RHPmoment}) as follows.
\begin{thm}
the generating function for the RHP reads\[
\sum_{n=0}^{+\infty}\frac{H_{n}^{N}\left(X\right)}{n!}t^{n}=\left(\left(1-\frac{tX}{N}\right)^{2}+\frac{t^{2}}{N}\right)^{-N}\]
for $\vert t\vert<\frac{\sqrt{N}}{\sqrt{1+\frac{X^{2}}{N}}}.$\end{thm}
\begin{proof}
Starting from (\ref{eq:momentHermiteUV}) we obtain

\begin{eqnarray*}
\sum_{n=0}^{+\infty}\frac{H_{n}^{N}\left(X\right)}{n!}t^{n} & = & E_{U,V}\sum_{n=0}^{+\infty}\frac{\left(\frac{t}{\sqrt{N}}\right)^{n}}{n!}\left[\left(i+\frac{X}{\sqrt{N}}\right)U+\left(-i+\frac{X}{\sqrt{N}}\right)V\right]^{n}\\
 & = & E_{U}\exp\left(\frac{t}{\sqrt{N}}\left(i+\frac{X}{\sqrt{N}}\right)U\right)E_{V}\exp\left(\frac{t}{\sqrt{N}}\left(-i+\frac{X}{\sqrt{N}}\right)V\right)\end{eqnarray*}
with $E_{U}\exp\left(\lambda U\right)=\left(1-\lambda\right)^{-N}$
for $\vert\lambda\vert<1$ so that, for $\vert t\vert<\frac{\sqrt{N}}{\sqrt{1+\frac{X^{2}}{N}}},$
\begin{eqnarray*}
\sum_{n=0}^{+\infty}\frac{H_{n}^{N}\left(X\right)}{n!}t^{n} & = & \left(1-\frac{t}{\sqrt{N}}\left(i+\frac{X}{\sqrt{N}}\right)\right)^{-N}\left(1-\frac{t}{\sqrt{N}}\left(-i+\frac{X}{\sqrt{N}}\right)\right)^{-N}\\
 & = & \left(\left(1-\frac{tX}{N}\right)^{2}+\frac{t^{2}}{N}\right)^{-N}\end{eqnarray*}
We remark that we recover the generating funtion for Hermite polynomials
as $N\to+\infty.$ The proof using the moment representation (\ref{eq:RHPmoment})
reads\begin{eqnarray*}
\sum_{n=0}^{+\infty}\frac{H_{n}^{N}\left(X\right)}{n!}t^{n} & = & E_{Z_{N}}\sum_{n=0}^{+\infty}\frac{\left(2N\right)_{n}}{n!}\left[\frac{t}{\sqrt{N}}\left(\frac{X}{\sqrt{N}}+iZ_{N}\right)\right]^{n}\\
 & = & E_{Z_{N}}\left(1-\frac{t}{\sqrt{N}}\left(\frac{X}{\sqrt{N}}+iZ_{N}\right)\right)^{-2N}\end{eqnarray*}
for $\vert t\vert<\frac{\sqrt{N}}{\sqrt{1+\frac{X^{2}}{N}}}.$ But
from \cite[3.665.1]{Gradshteyn}\[
E_{Z_{N}}\left(a-ibZ_{N}\right)^{-2N}=\frac{\Gamma\left(N+\frac{1}{2}\right)}{\Gamma\left(\frac{1}{2}\right)\Gamma\left(N\right)}\int_{0}^{\pi}\left(a-ib\cos x\right)^{-2N}\sin^{2N-1}xdx=\left(a^{2}+b^{2}\right)^{-N}\]
so that the result follows with $a=1-\frac{tX}{N}$ and $b=\frac{t}{\sqrt{N}}.$
\end{proof}
We note that formula \cite[3.665.1]{Gradshteyn} can be reovered using
only probabilistic tools as follows: since $\left(1-u\right)^{-2N}=E_{W_{2N}}\exp\left(uW\right)$
is the moment generating function of a Gamma random variable $W_{2N}$
with shape parameter $2N,$ we deduce \[
E_{Z_{N}}\left(1-bZ_{N}\right)^{-2N}=E_{Z_{N},W_{2N}}\exp\left(bW_{2N}Z_{N}\right).\]
But, by (\ref{eq:stochasticZN}), \[
W_{2N}Z_{N}=U-V\]
where $U$ and $V$ are two independent Gamma random variables with
shape parameter equal to $N.$ Thus\[
E_{W_{2N},Z_{N}}\exp\left(bW_{2N}Z_{N}\right)=E_{U,V}\exp\left(bU\right)\exp\left(-bV\right)=\left(1-b^{2}\right)^{-N}\]
and the result follows.

\subsection{Feldheim and Villenkin }

We give here a short proof of the Feldheim Villenkin generating function
for the normalized Gegenbauer polynomials\[
\sum_{n=0}^{+\infty}\frac{C_{n}^{N}\left(\cos\theta\right)}{C_{n}^{N}\left(1\right)}\frac{r^{n}}{n!}=\exp\left(r\cos\theta\right)j_{N-\frac{1}{2}}\left(r\sin\theta\right)\]
where the function $j_{N-\frac{1}{2}}$ is defined as in (\ref{eq:Besselj}),
by remarking that, using the moment representation (\ref{eq:momentCZN}),\[
\frac{C_{n}^{N}\left(\cos\theta\right)}{C_{n}^{N}\left(1\right)}=E_{Z_{N}}\left(\cos\theta+iZ_{N}\sin\theta\right)^{n}\]
so that\[
\sum_{n=0}^{+\infty}\frac{C_{n}^{N}\left(\cos\theta\right)}{C_{n}^{N}\left(1\right)}\frac{r^{n}}{n!}=E_{Z_{N}}\sum_{n=0}^{+\infty}\left(\cos\theta+iZ_{N}\sin\theta\right)^{n}\frac{r^{n}}{n!}=\exp\left(r\cos\theta\right)E_{Z_{N}}\exp\left(irZ_{N}\sin\theta\right)\]
The latest expectation is nothing but the characteristic function
$\phi_{Z_{N}}\left(u\right)=j_{N-\frac{1}{2}}\left(u\right)$ of $Z_{N}$
computed at $u=r\sin\theta,$ so that the result follows.

An equivalent result for the RHP is as follows:
\begin{thm}
A generating function for the relativistic Hermite polynomials is
\[
\sum_{n=0}^{+\infty}\frac{N^{\frac{n}{2}}}{\left(2N\right)_{n}}H_{n}^{N}\left(X\sqrt{N}\right)\frac{r^{n}}{n!}=\exp\left(rX\right)j_{N-\frac{1}{2}}\left(r\right).\]
\end{thm}
\begin{proof}
The proof follows the same line as the one above, starting from the
moment representation (\ref{eq:RHPmoment}).
\end{proof}

\subsection{Another generating function for the Hermite polynomials}

A classical generating function for the Hermite polynomials \cite[4.6.29]{Ismail}
reads\[
\sum_{n=0}^{+\infty}\frac{H_{n+k}\left(X\right)}{n!}t^{n}=\phi\left(X,t\right)H_{k}\left(X-t\right)\]
where $\phi\left(X,t\right)=\exp\left(2Xt-t^{2}\right)$ is the generating
function of the Hermite polynomials.

A generalization of this formula to the relativistic Hermite polynomials
reads as follows.
\begin{thm}
For the relativistic Hermite polynomials,\[
\sum_{n=0}^{+\infty}\frac{H_{n+k}^{N}\left(X\right)}{n!}t^{n}=\left(\phi_{N}\left(X,t\right)\right)^{1+\frac{k}{N}}H_{k}^{N}\left(X-\left(1+\frac{X^{2}}{N}\right)t\right)\]
where $\phi_{N}\left(X,t\right)=\left(1-\frac{2Xt}{N}+\frac{X^{2}t^{2}}{N^{2}}+\frac{t^{2}}{N^{2}}\right)^{-N}$
is the generating function of the relativistic Hermite polynomials.\end{thm}
\begin{proof}
Denote\[
f\left(X,t\right)=1-2\frac{Xt}{N}+\frac{X^{2}t^{2}}{N^{2}}+\frac{t^{2}}{N^{2}}\]
so that \[
\sum_{n=0}^{+\infty}\frac{H_{n}^{N}\left(X\right)}{n!}t^{n}=f^{-N}\left(X,t\right)=\phi_{N}\left(X,t\right)\]
and, using the moment representation (\ref{eq:momentHermiteUV}),\begin{eqnarray*}
\sum_{n=0}^{+\infty}\frac{H_{n+k}^{N}\left(X\right)}{n!}t^{n} & = & E_{U,V}\sum_{n=0}^{+\infty}\frac{1}{n!}\left(\frac{t}{\sqrt{N}}\right)^{n}\left[\left(i+\frac{X}{\sqrt{N}}\right)U+\left(-i+\frac{X}{\sqrt{N}}\right)V\right]^{n+k}\\
 & = & E_{U,V}\left[\left(i+\frac{X}{\sqrt{N}}\right)U+\left(-i+\frac{X}{\sqrt{N}}\right)V\right]^{k}\exp\left(\frac{t}{\sqrt{N}}\left[\left(i+\frac{X}{\sqrt{N}}\right)U+\left(-i+\frac{X}{\sqrt{N}}\right)V\right]\right)\\
 & = & \frac{d^{k}}{dt^{k}}\phi_{N}\left(X,t\right)\end{eqnarray*}
 Define $\beta=\sqrt{1+\frac{X^{2}}{N}}$ and $z=\beta^{2}\left(t-\frac{X}{\beta^{2}}\right)$
so that $\phi_{N}\left(X,t\right)=\beta^{2N}\left(1+\frac{z^{2}}{N}\right)^{-N}$
and\begin{eqnarray*}
\frac{d^{k}}{dt^{k}}\phi_{N}\left(X,t\right) & = & \beta^{2N}\frac{d^{k}}{dt^{k}}\left(1+\frac{z^{2}}{N}\right)^{-N}=\beta^{2N+2k}\left(-1\right)^{k}\left(1+\frac{z^{2}}{N}\right)^{-N-k}H_{k}^{N}\left(z\right)\\
 & = & \left(\frac{1+\frac{X^{2}}{N}}{1+\frac{\left(\left(1+\frac{X^{2}}{N}\right)t-X\right)^{2}}{N}}\right)^{N+k}H_{k}^{N}\left(\left(1+\frac{X^{2}}{N}\right)t-X\right)\\
 & = & f\left(X,t\right)^{-N-k}H_{k}^{N}\left(\left(1+\frac{X^{2}}{N}\right)t-X\right)\end{eqnarray*}
so that the result holds. 
\end{proof}

\section{Determinants}

Determinants with orthogonal polynomials entries have been extensively
studied \cite{Karlin} by Karlin and Szegö and recently revisited
by Ismail \cite{Ismail_det}. We show here that, in the case of Turàn
determinants, the moment representation derived above allows to extend
some of these results to relativistic Hermite polynomials.

We propose a method slightly different from the one used in \cite{Ismail_det}
based on the following result.
\begin{thm}
\label{thm:Dngeneralcase}If the polynomials $P_{n}\left(X\right)$
can be expressed as \[
P_{n}\left(X\right)=E_{Z}\left[X+iZ\right]^{n}\]
for some random variable $Z$ then the Turàn determinant\[
D_{n}^{P}\left(X\right)=\det\left[\begin{array}{ccc}
P_{0}\left(X\right) & \dots & P_{n}\left(X\right)\\
P_{1}\left(X\right) & \dots & P_{n+1}\left(X\right)\\
\vdots &  & \vdots\\
P_{n}\left(X\right) & \dots & P_{2n}\left(X\right)\end{array}\right]\]
is a constant equal to \[
D_{n}^{P}\left(X\right)=\frac{\left(-1\right)^{\frac{n\left(n+1\right)}{2}}}{\left(n+1\right)!}E_{Z_{0},\dots,Z_{n}}\prod_{0\le j<k\le n}\left(Z_{j}-Z_{k}\right)^{2}.\]
\end{thm}
\begin{proof}
We use the formula of Wilks \cite{Wilks}\[
\det\left[\begin{array}{ccc}
m_{0} & \dots & m_{n}\\
m_{1} & \dots & m_{n+1}\\
\vdots &  & \vdots\\
m_{n} & \dots & m_{2n}\end{array}\right]=\frac{1}{\left(n+1\right)!}E_{U_{0},\dots,U_{n}}\prod_{0\le j<k\le n}\left(U_{j}-U_{k}\right)^{2}\]
where $m_{k}=EU_{0}^{k}$ and the $U_{i}$ are independent and identically
distributed. Thus since \[
P_{n}\left(X\right)=E_{Z}U^{n}\]
with $U=X+iZ,$ we deduce \[
U_{j}-U_{k}=\left(X+iZ_{j}\right)-\left(X+iZ_{k}\right)=i\left(Z_{j}-Z_{k}\right)\]
so that \[
D_{n}^{P}\left(X\right)=\frac{\left(-1\right)^{\frac{n\left(n+1\right)}{2}}}{\left(n+1\right)!}E_{Z_{0},\dots,Z_{n}}\prod_{0\le j<k\le n}\left(Z_{j}-Z_{k}\right)^{2}\]
and the result follows. 
\end{proof}
Using this result, we deduce the Turàn determinant for the normalized
Relativistic Hermite polynomial defined as \begin{equation}
\mathcal{H}_{n}^{N}\left(X\right)=\frac{N^{\frac{n}{2}}}{\left(2N\right)_{n}}H_{n}^{N}\left(X\sqrt{N}\right).\label{eq:normalizedRHP}\end{equation}

\begin{thm}
The Turàn determinant for the normalized polynomial (\ref{eq:normalizedRHP})
is a constant equal to\[
D_{n}^{\mathcal{H}}\left(X\right)=\frac{\left(-1\right)^{\frac{n\left(n+1\right)}{2}}}{2^{n\left(n+1\right)}}\prod_{j=1}^{n}\frac{j!\left(2N-1\right)_{j}}{\left(N-\frac{1}{2}\right)_{j}\left(N+\frac{1}{2}\right)_{j}}.\]
\end{thm}
\begin{proof}
From Theorem \ref{thm:Dngeneralcase}, we have \[
D_{n}^{\mathcal{H}}\left(X\right)=\frac{\left(-1\right)^{\frac{n\left(n+1\right)}{2}}}{\left(n+1\right)!}E_{Z_{0},\dots,Z_{n}}\prod_{0\le j<k\le n}\left(Z_{j}-Z_{k}\right)^{2}\]
where, from (\ref{eq:RHPmoment}), each $Z_{j}$ is distributed according
to (\ref{eq:fZn}). This expectation is a Selberg integral equal to
\cite[17.6.2]{Mehta} \[
E_{Z_{0},\dots,Z_{n}}\prod_{0\le j<k\le n}\left(Z_{j}-Z_{k}\right)^{2}=\left(\frac{2^{n+2N-1}\Gamma\left(N+\frac{1}{2}\right)}{\Gamma\left(N\right)\Gamma\left(\frac{1}{2}\right)}\right)^{n+1}\prod_{j=0}^{n}\frac{\left(j+1\right)!\Gamma^{2}\left(N+j\right)}{\Gamma\left(2N+n+j\right)}\]
and the result follows after some elementary algebra.
\end{proof}
We note that this result can be proved as a consequence of theorem
5 by Ismail \cite{Ismail_det}: for normalized Gegenbauer polynomials
defined as \[
\mathcal{C}_{n}^{N}\left(X\right)=\frac{n!}{\left(2N\right)_{n}}C_{n}^{N}\left(X\right),\]
 the Turàn determinant equals\[
D_{n}^{\mathcal{C}}\left(X\right)=\left(\frac{X^{2}-1}{4}\right)^{\frac{n\left(n+1\right)}{2}}\prod_{j=1}^{n}\frac{j!\left(2N-1\right)_{j}}{\left(N-\frac{1}{2}\right)_{j}\left(N+\frac{1}{2}\right)_{j}}.\]
Applying Nagel's identity,\[
\mathcal{H}_{n}^{N}\left(X\right)=\left(1+X^{2}\right)^{\frac{n}{2}}\mathcal{C}_{n}^{N}\left(\frac{X}{\sqrt{1+X^{2}}}\right)\]
so that \[
D_{n}^{\mathcal{H}}\left(X\right)=\left(1+X^{2}\right)^{n\left(n+1\right)}D_{n}^{\mathcal{C}}\left(X\right)\left(\frac{X}{\sqrt{1+X^{2}}}\right).\]
Elementary algebra yields the result. We remark the similarity between
the former formula and Nagel identity.

\section{Conclusion}

Some new results about Relativistic Hermite polynomials have been
shown; the important fact is that several tools (subordination, moment
representation) have been used, depending on the type of the result.

\end{document}